\documentclass[12pt,leqno]{amsart}

\usepackage{amsmath}
\usepackage{amssymb}
\usepackage{a4wide}
\usepackage{enumitem}
\usepackage{graphics}
\usepackage{epsfig}
\usepackage{hyperref}
\hypersetup{%
  colorlinks = true,
  linkcolor  = blue,
  citecolor    = red
}
  \usepackage[applemac]{inputenc}
\usepackage{amsmath, amssymb, latexsym, amscd, amsthm,amsfonts,amstext}
\usepackage[mathscr]{eucal}
\usepackage{appendix}
\usepackage[active]{srcltx}
\include{srctex}

\usepackage{color}

 \usepackage{mleftright}
\usepackage{pgf,tikz}
\usepackage{caption}  
\usepackage{mathrsfs}
\usetikzlibrary{quotes, angles,decorations.pathreplacing}
\usepackage{comment}

\makeatletter
\@addtoreset{equation}{section}
 \makeatother
 
 \usepackage{hyperref}

\newtheorem{theorem}{Theorem}

\newtheorem{proposition}[theorem]{Proposition}

\newtheorem{remark}[theorem]{Remark}

 
\title[semiclassical asymptotics]
{A Gaussian Beam Construction of de  Haas-van Alfven Resonances }

\begin{document}

\author[M. Dimassi, J.-C. Guillot   \;and \; J. Ralston]{Mouez Dimassi, Jean-Claude Guillot  and James Ralston}
\address{Mouez Dimassi, IMB (UMR-CNRS 5251), UNIVERSIT\'E DE BORDEAUX, 351 COURS DE LA LIB\'ERATION,
33405 TALENCE CEDEX, FRANCE}
\email{mdimassi@u-bordeaux.fr}
\address{{ J.-C. Guillot, UNIVERSIT\'E  D'AIX-MARSEILLE, CNRS, INSTITUT DE MATH\'EMATIQUES DE MARSEILLE, UMR 7373, 13453 MARSEILLE CEDEX 13, FRANCE} }
\email{{ jcguillot@math.cnrs.fr}}
\address{J. Ralston UNIVERSITY OF CALIFORNIA, LOS ANGELES,   CA 90095, USA}
\email{ralston@math.ucla.edu}

\keywords{Schr\"odinger operator,  de Haas-van Alfven oscillations, gaussian beam construction,
magnetic field, generalized Onsager relation}
\subjclass[2000]{81Q10 (35P20 47A55 47N50 81Q15)}



\maketitle
\begin{abstract} 
In this article  we consider a Bloch electron in a crystal lattice subject to slowy varying external magnetic fields. We offer an explanation of de Haas-van Alfven oscillations in terms of
energy levels of approximate eigenfunctions for the magnetic Schr\"odinger operator by using a gaussian beam construction for a small enough magnetic field.

 
 
 
 

\end{abstract}

\section{Introduction}

The quantum dynamics of a Bloch electron in a crystal subject to external  {  { constant magnetic field $\nabla\times A$}} is governed by the Schr\"odinger equation
{ \begin{equation}\tag{1}\label{EQ00}
{ (P-E_0)u:= \Big[{\hbar^2\over 2m}(D_x+\mu A(x))^2 + eV({x})-E_0\Big] u =0},  \,\, D_x=\frac{1}{i}\partial_x,
\end{equation}}
where $V$  is a smooth, real-valued potential, periodic with respect to  a lattice $\Gamma=\oplus_{i=1}^3\mathbb Z a_i$ in $\mathbb R^3$. Here $(a_1,a_2,a_3)$ is a basis of $\mathbb R^3$,  $m$  and $e$  are the mass and charge of the electron, and { $\mu \hbar  =e$}. The magnetic potential $A(x)=(0, \epsilon x_1, 0)$ corresponding to the constant magnetic field 
$\nabla \times A = (0, 0, \epsilon)$. In this paper we will treat the magnetic field strength $\epsilon$ as a small parameter and use the scaled variable $y=\epsilon x$ and the potential $A(y) = (0, y_1, 0)$.


In the semi-classical dynamics of Bloch electrons under slowly varying electric and magnetic fields, recent advances have been made  (see \cite{BR, Bu, DGR1, SN} and the references given there).
Since the work of Peierls \cite{Pe}  and Slater \cite{Sl}, it is well known that, if $\epsilon$ is sufficiently small, the wave packets
   are propagating along  the trajectories from  the semi-classical Hamiltonian { $H(y,p)=E(p+\mu A(y))$, $y=\epsilon x$}. Here $E(k)$ is one of the band functions describing the Floquet spectrum of the unperturbed Hamiltonian  $-\frac{\hbar^2}{ 2m}\Delta+eV(x).$ {  In order $\epsilon$, the semiclassical quantization condition for magnetic levels (well-known Onsager relation), contains two phases :
One is the Berry's phase,  and the other is known as the Wilkinson-Rammal  phase (see \cite{DGR1,SN}). }

The orbit of the full  classical Hamiltonian  $H(y,p)$  is  helical  and cannot be quantized. Its projection on the pseudo-momentum coordinate $k=p+\mu A(y)$ lies in   the intersection of the Fermi surface $\{E(k)=E_0\}$  with the plane where  $k_3$ is constant. Under the assumption that this intersection  is   a  simple closed curve, 
 the electron's motion that is perpendicular to the magnetic field is quantized. In order to examine the generalized Onsager relation, we will employ the reduced Peierls classical Hamiltonian, $\hat H(y_1, y_2, p_1, p_2):=E(p_1,p_2+\mu y_1, k_3)$, where $k_3$ is constant and   the variable  $y_3$ is deleted (for more information, consult  sections 2-3).


\vskip.2in
In this work, we  use the gaussian beam (GB) to construct for $\epsilon$ small enough asymptotic solutions of \eqref{EQ00} concentrated in a tube of radius $\epsilon^{1/2}$ around the curve  $\Pi_{\hat y}{\hat \gamma}$ which is traced by $\hat y(s)=(y_1(s), y_2(s))$. Here, $\hat \gamma:=\{(\hat y(s), \hat p(s)), s\in [0, T]\}$ is a periodic  
 trajectory for the reduced  Hamiltonian $\hat H(\hat y, \hat p)=E(p_1,p_2+\mu y_1, k_3)$, (see Proposition \ref{Prop 1} and Theorem \ref{Prin1}).
 
 In section 4, we use these approximate asymptotic solutions to study the  generalized semiclassical quantization condition for cyclotron orbits.
Due to the wave function being single-valued along a closed orbit, the quantization condition  including the
Berry, Wilkinson-Rammal (WR),  and Maslov phases is established in Theorem \ref{Prin2}. It should be noted that  the Berry phase and WR phase both involve cell-periodic Bloch wave functions and cannot be derived from the zero-field energy spectrum alone. The geometry of the Bloch states has a crucial impact on the phase of magnetic oscillations.
 
When $\hat \gamma$ is a  stable periodic orbit for the bicharacteristic flow,  GB is  well-understood for a large class of partial differential equations (see \cite{Ba, LRT, R2}). In our case,  the curve $\hat\gamma$ is unstable (i.e., all the eigenvalues of  the linearized Poincare  are real). The construction of quasi-modes in \cite{Ba, R2} and elsewhere does not allow this. To remove this difficulty, we will examine and adapt more closely Ralston's approach to our case.


 

\vskip.2in
\section{Preliminaries}

\subsection* {\bf Equations of motion in Physical and Pseudo-momentum spaces}

Let $E_n(k)$
 be  one of the band functions describing the Floquet spectrum of the unperturbed Hamiltonian:
$$H_0(k)={\hbar^2\over 2m} (D_x+k)^2+eV(x) : L^2(\mathbb T) \rightarrow L^2(\mathbb T).$$
Let 
$\Phi_n(\cdot,k)=e^{-ix\cdot k} \Psi_n(x,k)$ be the corresponding  normalized eigenfunction, 
  \begin{equation}\tag{2}\label{EF0}
 \Big[H_0(k)-E_n(k)\Big] \Phi_n(x,k)=0,\,\,\, \int_{\mathbb  T} \vert \Phi_n(x,k)\vert^2dx=1,
 \end{equation}
 where $\Psi_n(\cdot,k)$ is  the Bloch function    associated to $E_n(k)$ :
 $$\Psi_n(x+\gamma,k)=e^{ik\cdot \gamma}\Psi_n(x,k),\, \forall \gamma \in \Gamma.$$ 
 Since  $e^{-ix\cdot \gamma^*} H_0(k) e^{ix\cdot \gamma^*}=H_0( k+\gamma^*)$, it follows that
  that $$E_n(k+\gamma^*)=E_n(k),\,\,\,\,\, \text  { for all } \gamma^*\in \Gamma^*,$$
 where $\Gamma^*$ is  the  reciprocal lattice. {Standard perturbation theory shows that the function $E_n(k)$ is continuous for $k\in \mathbb R^3$ and real analytic
in a neighborhood of any $k$ such that $E_n(k)$ is a simple eigenvalue, i.e., 
$$E_{n-1}(k)<E_{n}(k)<E_{n+1}(k).$$
 For  $E_0\in E_n(\mathbb T^*)$, we put
${\mathcal F}(E_0)=\{k\in \mathbb T^*:   \, E_n(k)=E_0\}$ \footnote{ When $E_0$ equals  the Fermi energy $E_F$, ${\mathcal F}(E_0)$   is part of the Fermi surface defined by  ${\mathcal F}_F:=\{ k\in \mathbb T^*;  E_F\in \sigma(H_0(k))\}$ (see \cite{LK}). Here $\sigma(H_0(k))$ denotes the spectrum of the operator $H_0(k)$.}. We assume that for every $k=(k_1, k_2, k_3)\in {\mathcal F}(E_0)$ with $k_3$ in an open interval, $E_n(k)$ is a simple eigenvalue of $H_0(k)$. Therefore, $k\mapsto E_n(k) $ is analytic  in a neighborhood of ${\mathcal F}(E_0)$, and 
we can  choose  $k\mapsto \Psi_{n_0}(\cdot,k)$ to be a real-analytic function near ${\mathcal F}(E_0)$.
Since we will use only
one band, we will suppress the index $n$ in  $E_{n}(k)$, $\Phi_{n}(\cdot, k)$ and $\Psi_{n}(\cdot,k)$. }

In classical discussions (\cite{o}) the de Haas-van Alfven effect is associated with the curves in pseudo-momentum space obtained as intersections of planes perpendicular to the (constant) magnetic field with the Fermi surface. Here we use the curves in 
$(y_1, y_2, p_1, p_2)$ corresponding to electron paths. These curves are
determined by the Peierls Hamiltonian $H(y, p) = E(p + \mu A)$. However, they are obtained from the pseudo-momentum space curves above by substituting $k_1 = p_1, k_2 = p_2 + \mu y_1$    and $k_3 = p_3$, and setting $y_2 =-\frac{p_2}{\mu}$
plus a constant (our constructions will be independent of $y_3$). To verify this we compute as follows. Let
$(y(s), p(s))$ be a trajectory for the Peierls hamiltonian
generated by $H(y,p)$ :
 \begin{equation}\tag{3}\label{EM}
 \dot y(s)=\frac{\partial H}{\partial p}=\frac{\partial E}{\partial k}(p(s)+\mu A(y(s))),
 \end{equation}

 \begin{equation}\tag{4}\label{EM0}
 \dot p(s)=-\frac{\partial H}{\partial y}=-\mu (\dot y_2(s), 0, 0). 
 \end{equation}
 {In the pseudo-momentum coordinate $k(s):=p(s)+\mu A(y(s))$ one has $$E(k(s))=E(k_1(s), k_2(s), k_3)=E(k(0))=E_0,$$
 and
 \begin{equation}\tag{5}\label{EM1}
 \dot y(s)=\frac{\partial E}{\partial k}(k(s)),\,\,  \dot k(s)=\mu (- \dot y_2(s), \dot y_1(s), 0).
 \end{equation}
 
  Let $z_1=p_2+\mu y_1$. Since $\dot p_2=0$, we have 
   $$
 \dot z_1=\mu \frac{\partial E}{\partial k_1}(p_1,z_1,p_3) \text { and }  \dot p_1=-\mu\frac{ \partial E}{\partial k_2}(p_1,z_1,p_3).
 $$
  Hence,  since $\dot p_3=0$,  $(z_1(s),p_1(s))$  moves along a level curve for $E(p_1,z_1,k_3)$.
 We also have
  \begin{equation}\tag{6}\label{EM2}
 \dot y_2=\partial_{p_2}E(p_1,p_2+\mu y_1,p_3)=-\dot p_1/\mu.
 \end{equation}
Thus there are two families of trajectories here. First there the level curves
$$
\gamma=\gamma(k_3,E_0)=\{(k_1,k_2): E(k_1,k_2,k_3)=E_0\}.
$$
Second there are the trajectories (with $y_3$ deleted) for the Peierls Hamiltonian
  \begin{equation}\tag{7}\label{XXX00}
\hat \gamma=\hat \gamma(k_3,E_0,c)=\{(y_1,y_2,p_1,p_2):E(p_1,p_2+\mu y_1,k_3)=E_0,\  y_2=-p_1/\mu +c\}.
\end{equation}
Note that $\gamma (k_3,E_0)$ will become $\hat \gamma(k_3,E_0, c)$ when one substitutes $k_1=p_1$ and $k_2=p_2+\mu y_1$ and sets $y_2=-p_1/\mu +c$. $\hat \gamma$ will be used in the  sections Generalized Onsager Relation.

 


Throughout this paper we assume that :  
  \begin{equation}\tag{H1}\label{ASSU}
(\frac{\partial E}{\partial k_1},  \frac{\partial E}{\partial k_2})\not =(0,0) \text { on } {\gamma}(k_3, E_0),
 \end{equation}
 and
   \begin{equation}\tag{H2}\label{ASSU2}
 {\gamma}(k_3, E_0) \text { is a simple closed curve.}
 \end{equation}
Assumptions \eqref{ASSU} and \eqref{ASSU2} insure that  $\{(k_1, k_2) : E(k_1,k_2, t_3)=E\}$  are simple and closed for $\vert E-E_0\vert $ and $\vert t_3-k_3\vert$ small enough and that it depends smoothly on $(t_3, E)$.

 We let $S(k_3)$  denote the area in $k$-space   enclosed by $\gamma(k_3, E_0)$.   $\gamma(k_3, E_0)$ is the projection of a helical orbit $\hat\Gamma(k_3)$ of the
full hamiltonian in  \eqref{EM}, \eqref{EM0}.}

\section{Gaussian Beam Construction}

\subsection {Two scale expansions method, eikonal and transports equations}

In what follows, $\hat y$ (resp. $\hat x$, $\hat k$) denotes $(y_1, y_2)$ (resp. $(x_1, x_2)$, $ \hat k=(k_1, k_2)$).
With the change of variable $\hat y=\epsilon \hat x$, the operator $P$  is unitarily equivalent to
 $$\tilde P={\hbar^2\over 2m}\Big[(\epsilon D_{\hat y}+\mu A(\hat y))^2+D_{x_3}^2\Big] + eV(\frac{ \hat y}{\epsilon}, x_3).$$
 Here we  are looking  for a  solution 
 of   the equation 
 \begin{equation}\tag{8}\label{Eq}
 { (\tilde P-E_0)u(\hat y, x_3, \epsilon)\equiv 0}.
 \end{equation}
 of the form
$$u(\hat y,x_3, \epsilon)=e^ {{i}(\phi(\hat y)/\epsilon +x_3 k_3)} m(\frac{\hat y}{\epsilon }, x_3,  \hat  y; \epsilon).$$
  In order to accomplish this we use, as in \cite {Bu, DGR1, GRT},  the two-scale expansion method in which the  coordinate $\hat x$ and the slowly
varying
 space variable $\hat y=\epsilon \hat x$ are 
 regarded as independent variables. Thus, we  consider the following equation in the independent variables
$\hat x$ and $\hat y$ : 
 \begin{equation}\tag{9}\label{XX0}
({\mathbb P}-E_0)v=0,
\end{equation}
with
 \begin{equation}\tag{10}\label{XX01}
\mathbb P={\hbar^2\over 2m}\Big[(\epsilon D_{\hat y}+D_{\hat x}+\mu A(\hat y))^2+(D_{x_3}+k_3)^2\Big] + eV(\hat x, x_3)
\end{equation}
Note that,
if  
 $v(\hat x, x_3, \hat y,\epsilon)$ is a solution of \eqref{XX0}, 
   then  $u=e^{ix_3 k_3}v(\frac{\hat x}{\epsilon}, x_3, \hat y,\epsilon)$ is a 
solution of \eqref{Eq}. In the variable $x=(\hat x, x_3)$, $v(x, \hat y, \epsilon)$ is required 
to be periodic.

We look  for approximate solution to \eqref{XX0}, which have the form :
 \begin{equation}\tag{11}\label{XX02}
v(\hat x, x_3, \hat y, \epsilon)=
e^{i\phi(\hat y)/\epsilon} \Bigl[m_0(\hat x, x_3, \hat y)+\epsilon m_1(\hat x, x_3, \hat y)+\cdots+\epsilon^Nm_N(\hat x, x_3, \hat y)\Bigr].
\end{equation}
Now substituting \eqref{XX02}  into \eqref{XX0} and collecting terms which are the same order in $\epsilon$, we get
\begin{equation}\tag{12}\label{XX03}
(\mathbb P-E_0)v=e^{i\phi(\hat y)/\epsilon} \Bigl[c_0(\hat x, x_3, \hat y)+\epsilon c_1(\hat x, x_3, \hat y)+\cdots+\epsilon^{N+2}c_{N+2}(\hat x, x_3, \hat y)\Bigr] \end{equation}
where 
\begin{equation}\tag{13}\label{XX04}
 c_0(\hat x, x_3, \hat y)=\left[ H_0\left(K( \hat y)\right)
-E_0\right]m_0,\end{equation}
\begin{equation}\tag{14}\label{XX05}
c_1(\hat x, x_3, \hat y)= \left[ H_0
\left( K(\hat y)\right)-E_0\right]m_1-\frac{\hbar^2}{2m}\mathbb Km_0,\end{equation}
and for $j=2,3,...,N+2$
$$
c_j(\hat x, x_3, \hat y)=\left[ H_0
\left(K(\hat y)\right)-E_0\right]m_j-\frac{\hbar^2}{2m} \Big(\mathbb K m_{j-1}+\Delta_{\hat y} m_{j-2}\Big).$$
Here
$$\mathbb K=i\Big[
{\partial H_0\over \partial \hat k}( K(\hat y))\cdot{\partial \over \partial \hat y}+
\Delta_{\hat y} \phi\Big]$$
and 
$$  K(\hat y)=({\partial \phi
\over\partial \hat y}(\hat y)+\mu A(\hat y), k_3).$$

When $\phi$ is real-valued, \eqref{XX02}  is the standard ansatz of the WKB-method. In this case, to solve the equation \eqref{XX0}, one requires that $c_j=0, j=0, 1, 2, \cdots$ :

 \begin{equation}\tag{E}\label{Eik}
c_0=\left[ H_0\left(K( \hat y)\right)
-E_0\right]m_0=0, \: \:\:\ ({\rm eikonal \:   equation})
\end{equation}
\begin{equation}\tag{$T_1$}\label{TraT1}
 c_1=\left[ H_0
\left( K(\hat y)\right)-E_0\right]m_1-\frac{\hbar^2}{2m}\mathbb Km_0=0,\,\,\, ({\rm transport \: equation }\,\, T_1)
\end{equation}
 \begin{equation}\tag{$T_j$}\label{Tran}
  c_j(\hat x, x_3, \hat y)=0, \:\: ({\rm transport \: equation } \,\, T_j). 
  \end{equation}
  
  \subsection{ Construction of the phase function by the Gaussian beam method}
 
  According to \eqref{Eik}, equation  tells us that for all $\hat y$, $m_0(\hat x, x_3,  \hat y) $ is an eigenfunction of  $H_0(K(\hat y))$
with eigenvalue $E_0$. Therefore, we can fulfill \eqref{Eik} by choosing
 \begin{equation}\tag{$E_\phi$}\label{Eikon}
E({\partial \phi
\over\partial \hat y}(\hat y)+\mu A(\hat y), k_3)=E_0,
\end{equation}
$$m_0(\hat x, x_3, \hat y)=f_0(\hat y) \Phi(x, K(\hat y)).$$
The phase $\phi$ is derived from the  classical hamiltonian
$$\widehat H(\hat y,  \hat p):=\widehat H(y_1, y_2, p_1, p_2)=E(p_1, p_2+\mu y_1, k_3).$$

By definition (see \eqref{XXX00}),  $ \hat\gamma$ is a periodic orbit for the Hamiltonian system
\begin{equation}\tag{15-0}\label{EQ0}
\dot y=\widehat H_p,\ \dot p=-\widehat H_y.
\end{equation}
Along $\hat\gamma$,   $\phi$ satisfies $(p_1(s), p_2(s))=\phi_y(y(s))$ with
\begin{equation}\tag{15-1}\label{EQ}
E(k(s))=E_0,\,\,\, k(s)=(p_1(s), p_2(s)+\mu y_1(s),k_3)=(\phi_{\hat y}(\hat y(s))+\mu A(\hat y(s)),k_3).
\end{equation}

 The nonlinearity of the Hamilton-Jacobi equation \eqref{Eikon} generally leads to finite time singularity formation in phase $\phi$ "caustic problem", and the transport equations then become undefined. The caustic problem has been addressed in many works, starting with Keller, Maslov, and H\"ormander, using the classical Fourier integral operator (FIO) approach \cite{Ho}. Us indicated in the introduction, instead of OIF we are going to use "Gaussian Beams" as in \cite{Ar, R1, R2}. This means that  we are not going to attempt to solve \eqref{Eikon}  exactly,  we only need to  build asymptotic solutions concentrated on a single ray $\hat \gamma$. The Gaussian profile is achieved by allowing the phase to be complex away from the ray so that the solution decays exponentially away from $\hat \gamma$.

 More precisely, fix  a periodic trajectory $\hat \gamma=\{(\hat y(s), \hat p(s));  s\in [0, T]\}$ of  the classical hamiltonian $\hat H$ with initial date $(\hat y(0), \hat p(0))$ satisfying $\hat H (\hat y(0), \hat p(0))=E_0$, and let $ \Pi_{\hat y} \hat \gamma$  denote  the projection of  $\hat \gamma$ on the $\hat y$-space. We are going to prove 

\begin{proposition}\label{Prop 1}
Assume \eqref{ASSU} and \eqref{ASSU2}.  There exists a smooth function $\phi$ such that
\begin{enumerate}
\item $\frac{\partial \phi}{\partial \hat y}( \hat y(s))=\hat p(s)$.
\item In a small neighborhood $\Omega$ of $\Pi_{\hat y}\hat \gamma$  \begin{equation}\tag{15}\label{XX00012}
   G(y):=E({\partial \phi
\over\partial \hat y}(\hat y)+\mu A(\hat y), k_3)-E_0={\mathcal O}_N\Big(d(\hat y, \Pi_{\hat y}\hat \gamma)^3\Big), 
   \end{equation}
\item $\Im \phi\geq C d(\hat y,   \Pi_{\hat y} \hat \gamma)^2$,
\end{enumerate}
where $d(\hat y,   \Pi_{\hat y} \hat \gamma)$ is the  distance from $\hat y$ to $ \Pi_{\hat y} \hat \gamma$.
\end{proposition}
  \begin{proof}
  
  The proof is adapted from  \cite{R2}.  For this reason we omit some details. 
  Here,  the main difficulty in carrying out the construction in \cite{R2} is that
 all orbits for $\widehat H$ near $\hat\gamma$ are periodic,  hence the algebraic eigenvalues of 
  the linearized Poincare map $\mathcal P$
   are all $1$. This is not allowed in the construction of quasi-modes in \cite{R2} and elsewhere. 
   
   Recall that $k_3$  is fixed, and  the classical hamiltonian $\widehat H(y_1, y_2, p_1, p_2)=E(p_1, p_2+\mu y_1, k_3)$ is independent on $y_3$. By abuse of notation, we write $y$, $y(s)$,  $p$, $H$ instead of $\hat y=(y_1,y_2)$, $\hat y(s)=(y_1(s), y_2(s))$, $\hat p=(p_1, p_2)$,  and $\widehat H$.
   
   Requiring that $G(y)=H(y, {\partial \phi\over\partial  y}( y))-E_0$ vanishes to zero, first  and second order on $\hat \gamma$, we get :
   \begin{enumerate}[label={\arabic*.}]
   \item  $\frac{\partial \phi}{\partial y}( y(s))= p(s)$,
   \item $\frac{ \partial G}{\partial y_j}=\frac{\partial H}{\partial y_j}+\sum_{l=1,2} \frac{\partial H}{\partial p_l} \frac{\partial^2\phi}{\partial_{y_j}\partial_{y_l}}=0$, for $j=1,2$,
\item   $\frac{ \partial ^2G}{\partial y_i\partial y_j}=\frac{\partial^2 H}{\partial y_i\partial y_j}+\sum_{l=1,2} \Big(\frac{\partial^2\phi}{\partial_{y_i}\partial_{y_l}}\frac{\partial^2H}{\partial p_l\partial y_j}
+\frac{\partial^2H}{\partial y_i\partial p_l}
\frac{\partial^2\phi}{\partial_{y_l}\partial_{y_j}}+\frac{\partial H}{\partial p_l}\frac{\partial^3\phi}{\partial y_l\partial y_i \partial y_j}\Big)=0$ for $i,j=1,2$,
 where the equalities are evaluated    along $\hat \gamma$.  
  \end{enumerate}
  
 It's worth noting that  1. is a consequence  of the  uniqueness of the solution of \eqref{EQ0} with  initial condition. While  2. is only the compatibility condition $\dot p(s)=\frac{d}{ds}\Big(\frac{\partial \phi}{\partial y}(y(s))\Big)$. Let us investigate the third condition.
 
 Introducing the matrix $$\left(M(s)\right)_{ij}=\frac{\partial^2\phi}{\partial_{y_i}\partial_{y_j}}(y(s)), \,\,\,\,\,\,\,\,\,\, \left(A(s)\right)_{ij}=\frac{\partial^2 H}{\partial_{y_i}\partial_{y_j}}(y(s), p(s)),$$
 $$
 \left(B(s)\right)_{ij}=\frac{\partial^2 H}{\partial_{p_i}\partial_{y_j}}(y(s), p(s)), \,\,\,\,\,\,\,\,\,\,\,\,\,\,\left(C(s)\right)_{ij}=\frac{\partial^2H}{\partial_{p_i}\partial_{p_j}}(y(s), p(s)),$$
   and using \eqref{EQ0}, one can rewrite the condition 3)  as the non-linear Ricatti matrix equation
     \begin{equation}\tag{RE}\label{Ricatti}
   \frac{dM}{d s}+MCM+MB+B^TM+A=0.
   \end{equation}
    
   Thus in order to prove Proposition 1,  we need to construct a phase $\phi$ such the  the matrix $M(s)$ satisfies \eqref{Ricatti} for all $s$ with
   \begin{itemize}
   \item $M(s)^T=M(s)$,
   \item $M(s+T)=M(s)$,
   \item $M(s) \dot y(s)=\dot p(s)$,
   \item  $\Im M(s)$ is positive definite on the orthogonal complement of $\dot y(s)$,
   
   \end{itemize}
   Some well known facts about the non-linear Ricatti equation \eqref{Ricatti}  are recalled in an appendix.

     Next, let us introduce the linearized equation about $\hat \gamma$ 
   \begin{equation}\tag{16}\label{A3}
 \left \{
\begin{array}{lr}
\dot {\delta y}= C(s) \delta  p +B(s) \delta y\\
\\
\dot{\delta p}=-B^T(s)\delta p- A(s)\delta y,
\end{array}
\right .
\end{equation}
 and the linearized Poincar\'e map $\mathcal P$ taking the data of solutions to \eqref{A3} at $s=0$ to their data at $s=T$, (i.e., ${\mathcal P} : (\delta y(0), \delta p (0))\rightarrow (\delta  y(T), \delta p(T))$).

To  apply the arguments from \cite{R2},  we need two vector solutions of \eqref{A3}, $v_1(s)$ and $v_2(s)$, where $v_1(s)$ is the tangent to $\hat \gamma$, i.e.  $v_1=(\dot y_
1,\dot y_2, \dot p_1,\dot p_2 (=0))$ which will satisfy \eqref{A3}  because it is the derivative of the flow  of $ H$ with respect to $s$.

Since we have assumed that $v_1(s)$ is never zero, $v_2(s)$ must satisfy three conditions. Letting $\sigma ((y,\eta),(w,\zeta))= y\cdot \zeta - w\cdot \eta$ be the symplectic two form, we need :
\begin{enumerate}
\item
  $ \sigma(v_2(s),v_1(s)) =0$ for all $s$.
\item 
$\sigma(v_2(s),\overline{v_2(s)})=ic$ with $c> 0$ for all $s$.
\item The complex span of $v_1(s)$ and $v_2(s)$, $S(s)$ should be periodic, i.e. $S(0)=S(T)$.
\end{enumerate}

Since $\sigma$ is constant on pairs of solutions to \eqref{A3} (see \eqref{A34}), the equalities i) and ii) will hold for all $s$, if they hold for $s=0$. Condition iii) makes this construction possible in our special case. The fact that all orbits near $\hat \gamma=\hat\gamma(E_0, k_3)$   are periodic means that for any solution $v(s)$ of \eqref{A3}, $\mathcal P$ maps $v(0)$ to $v(0)+a v_1(0)$ for some $a$. 
To see this choose curves $w_1(s)$ and $w_2(s)$ in $(y,p)$-space such that $w_1(0)=w_2(0)$ is on $\gamma$, and $\dot w_1(0)+i\dot w_2(0)=v_2(0)$.  Recalling that the
level surface  $\{(y, p), \, \hat H(y,p)=E\}$ is $3$-dimensional for  $\vert E-E_0\vert$ small enough, due
to the assumption \eqref{ASSU}. 
Since all orbits are periodic,   $w_1(s)$ and $w_2(s)$ are on periodic orbits,  $w_1(t,s)$ and $w_2(t,s)$ respectively. Note that $w_1(t,0)$ and $w_2(t,0)$ parametrize $\gamma$ and hence there are functions $ T_1(s) $ and $T_2(s)$ such that
\begin{equation}\tag{16-0}\label{GE0}
w_1(T_1(s),s)=w_1(0,s) \hbox{ and } w_2(T_2(s),s)=w_2(0,s).
\end{equation}

Evaluating the derivatives of the equations in \eqref{GE0}  with respect to $s$ at $s=0$, we have (since $T_j(0)=T$)

$$\partial_s w_j(0,0)=\partial_s w_j(T,0)+\partial_s T_j(0)\partial_t w_j(T,0),\  j=1,2, $$
which yields  $v_2(0)=v_2(T)+(\partial_sT_1(0)+i\partial_s T_2(0))v_1(T)$. 

Therefore,  $S(T)=S(0)$ for any choice of $v_2(s)$, and we  has plenty of freedom to choose $v_2(0)$ so that i) and ii) hold.
For instance, we could choose $(y(0),p(0))$ so that $\dot y_1(0)\neq 0$, and let  $v_2(s)=(\delta y_1^2(s), \delta y_2^2(s), \delta p_1^2(s), \delta p_2^2(s))$ with $v_2(0)=(0,i \dot y_1(0),-\dot y_2(0), \dot y_1(0))$. Let us introduce the $2\times 2$  matrices
$$Y(s):=(Y^1(s), Y^2(s))= \left(
\begin{matrix}
\dot y_1(s)& \delta y_1^2(s)\\
\dot y_2(s) & \delta y_2^2(s)
\end{matrix}
\right),
 N(s):=(N^1(s), N^2(s))= \left(
\begin{matrix}
\dot p_1(s) & \delta p_1^2(s)\\
\dot p_2(s)  & \delta p_2^2(s)
\end{matrix}
\right).
$$
Since we assumed that $(\dot y_1(s),\dot  y_2(s))$ never vanishes (see \eqref{EM} and \eqref{ASSU}),  we conclude from
Proposition \ref{Prop A1}  that $Y(s)$ is invertible for all $s$, hence that $M(s)=N(s)Y(s)^{-1}$ is well defined for all $s$, and finally that $M(s)$ is a global solution of \eqref{Ricatti}.

Now,  the desired properties of  the matrix $M(s)$ are derived  from Proposition \ref{Prop A1} and the fact that $S(0)=S(T)$.
This completes the  proof of Proposition \ref{Prop 1}. 
    \end{proof}

  \subsection{Construction of the Amplitude}

  Constructing the principal term $m_0$ is all we need to do for the applications. We offer suggestions on how to create the other terms in a remark.
  
 In the following, we assume that $\phi$ has been chosen so that Proposition\ref{Prop 1} holds, and we let
 $m_0=f_0(y)\Phi(x,K(y))$. Therefore
 \begin{equation}\tag{17}\label{XXXXX1}
 c_0=f_0(y)G(y)\Phi(x, K(y)).
     \end{equation}
     We want to choose $m_0, m_1,\cdots$, so that the functions $y\rightarrow c_j(x,y)$ vanish on $\hat \gamma$ to given order uniformly on $x\in \mathbb T$. 
 
  First,  let us deal with $c_0$ and  $c_1$.
  For  simplicity of notation, we write $E, H_0$ and $\Phi$ instead of $E(K(\hat y)), H_0(K(\hat y))$ and $\Phi(x, K(\hat y))$ respectively.
  
  Since $x\rightarrow m_1(x,\cdot)$ is required to be $\Gamma$-periodic, it is natural to write
  \begin{equation}\tag{18}\label{XXXX05}
  m_1(x,y)=f_1(y)\Phi(x, K(y))+m_1^\perp(x,y),
  \end{equation}
  where
  $$\langle \Phi(\cdot , K(y)), m^\perp_1(\cdot, y)\rangle=0.$$
Substituting the above equalities into \eqref{XX05} and using \eqref{EF0} we obtain
\begin{equation}\tag{19}\label{XX05-5}
c_1= f_1(y) G(y) \Phi(x, K(y))-\frac{\hbar^2}{2m}\mathbb Km_0+\left[ H_0
\left( K(y)\right)-E_0\right]m_1^\perp.\end{equation}
 By the Fredholm alternative, the equation,  $c_1=0$,  is solvable for   $m_1$ if and only if
 $$ f_1(y) G(y) \Phi(\cdot, K(y))-\frac{\hbar^2}{2m}\mathbb Km_0\perp {\rm ker} \left[ H_0-E_0\right] \text { in } 
 L^2(\mathbb T),$$
 for all $y$, 
  i.e.,
 \begin{equation}\tag{20}\label{XX00007}
 iL(y):=\Big <\mathbb K(f_0\Phi), \Phi\Big>_{L^2(\mathbb T)}-\frac{2m}{\hbar^2} f_1(y)G(y)=0.
 \end{equation}
 In view of the definition of $\mathbb K$,  we have
 \begin{equation}\tag{21}\label{XX0004}
L(y)=\frac{\partial f_0}{\partial  y}\cdot \Big<
{\partial H_0\over \partial k} \Phi, \Phi\Big>+ b(y) f_0(y)-\frac{2m}{i\hbar^2} f_1(y)G(y),
\end{equation}
where
 \begin{equation}\tag{22}\label{XX00004}
b( y):=\Big<
{\partial H_0\over \partial  k} \cdot \frac{\partial \Phi}{\partial  y}, \Phi\Big>+\Delta_{ y} \phi.
\end{equation}
We conclude from  \eqref{EF0} that
 \begin{equation}\tag{23}\label{XX004}
\frac{\partial E}{\partial  k}\Phi+(E-H_0)\frac{\partial \Phi}{\partial  k}
=\frac{\partial H_0}{\partial  k}\Phi,\,\,\, 
 \end{equation}
 hence that
$$  \frac{\partial E}{\partial  k}=\Big<\frac{\partial H_0}{\partial k} \phi, \phi\Big>.$$
Differentiating the above  equality  with respect to $ y$, and noting that $\frac{\partial H_0}{\partial k}$ is self-adjoint, we get
 \begin{equation}\tag{24}\label{XX005}
\frac{ \partial}{\partial  y}\cdot   \frac{\partial E}{\partial  k}=\Big<\frac{\partial H_0}{\partial  k}  \frac{\partial \phi}{\partial  y}, \phi\Big>+\Big<\frac{\partial H_0}{\partial  k} \phi, \frac{\partial \phi}{\partial  y}\Big>+2\Delta_{ y}\phi =2\Re b( y).
\end{equation}
Inserting \eqref{XX004} in \eqref{XX00004} we obtain
  \begin{equation}\tag{25}\label{XX0006}
\Im b=\Im\Big\{\frac{\partial E}{\partial  k}\cdot \Big< \frac{\partial \Phi}{\partial  y}, \Phi\Big>\Big\}+\Im\Big\{\Big< \frac{\partial \Phi}{\partial  y},(E-H_0)\frac{\partial \Phi}{\partial  k}\Big>\Big\}.
\end{equation}
Using the fact that $ \Big<\frac{\partial \Phi}{\partial  k_j},(E-H_0)\frac{\partial \Phi}{\partial  k_j}\Big>$ is real,  as well as the fact that $\frac{\partial K_2}{\partial  y_1}-\frac{\partial K_1}{\partial y_2}=\mu$,  we deduce that
 $$\Im\Big\{\Big< \frac{\partial \Phi}{\partial  y},(E-H_0)\frac{\partial \Phi}{\partial  k}\Big>\Big\}
= \mu \Im  \Big<(H_0-E)\frac{\partial \Phi}{\partial  k_1}, \frac{\partial \Phi}{\partial  k_2}\Big>.$$
  On the other hand, the normalization of $\Phi$  ensures that $\frac{1}{i}\Big< \frac{\partial \Phi}{\partial  y}, \Phi\Big>$ is real. This,  together with \eqref{XX0004}, \eqref{XX005} and \eqref{XX0006} yields
    \begin{equation}\tag{26}\label{XX0009}
L(y)=\frac{\partial E}{\partial  k}\cdot \frac{\partial f_0}{\partial  y}+ {\mathcal A}(y) f_0( y)
-\frac{2m}{i\hbar^2} f_1(y)G(y),
\end{equation}
where
$${\mathcal A}(y):=\frac{1}{2} \frac{ \partial}{\partial  y}\cdot   \frac{\partial E}{\partial k}+i\Big(\mu \Im  \Big<\frac{\partial \Phi}{\partial  k_2},(E-H_0)\frac{\partial \Phi}{\partial  k_1}\Big>+
\frac{1}{i}\frac{\partial E}{\partial  k}\cdot\Big< \frac{\partial \Phi}{\partial  y}, \Phi\Big>\Big).
$$
It is worth recalling that we want to choose $f_0$ and $f_1$ such  that the right hand side of \eqref{XX05-5} vanishes on $\hat\gamma$ to given  order. For any multi-index $\alpha$ of length $l$ the equations $\partial_y^\alpha  L(y)=0$ along $\hat\gamma$ 
gives rise to linear ordinary differential equation 
for $\partial_y^\alpha f_0$ with  inhomogeneous terms depending  on  the derivatives  of $f_0$ and $f_1$ of  order  up to $l-1$ and $l-3$ respectively :
 \begin{equation}\tag{27}\label{XXXX1}
\frac{\partial E}{\partial  k}\cdot \frac{\partial }{\partial  y}(\partial_y^{\alpha}f_0)+ {\mathcal A}(y) \partial_y^\alpha f_0( y)+C(f_0,\cdots,\partial_y^{\beta}f_0, f_1,\cdots\partial_y^{\beta'}f_1)
=0,
\end{equation}
with $\vert \beta'\vert\leq l-3$, $ \vert \beta\vert\leq l-1$  and  $C$ is independent on $f_1$ for $l<3$ and $C=0$ for $\alpha=0$.
Here we have used \eqref{XX00012}. 

Since $\dot y(s)=\frac{\partial E}{\partial  k}$, it follows from \eqref{XXXX1} and  the definition of $A(y)$ that :
   \begin{equation}\tag{28}\label{XX0007}
   \frac{d}{d s} \Big [(\partial_y ^\alpha f_0)(y(s))\Big]+\frac{1}{2}\Big(\partial_{ y}\cdot \partial_{ k}E(\phi_{y}+\mu A( y), k_3)\Big)_{\vert  y(s)} (\partial_y^\alpha f_0)(y(s))+i\Big(\dot \theta_b+\dot\theta_{rw}\Big)(\partial_y^\alpha f_0)( y(s))
   \end{equation}
 $$+C(f_0,\partial_yf_0,\cdots,\partial_y^{\beta}f_0, f_1,\cdots\partial_y^{\beta'} f_1)_{\vert_{y=y(s)}}=0,$$
 where, the phase $\theta_b$ is known as the Berry phase and $\theta_{rw}$ is the Wilkinson-Rammal  phase :
 \begin{equation}\tag{29}\label{Phases}
\dot \theta_b =i\Big<\Phi(\cdot, k(s)), \dot \Phi(\cdot,k(s))\Big>, 
\end{equation}
and
 \begin{equation}\tag{30}\label{RW}
\dot\theta_{rw}=\Im  \Big< (H_0(k(s))-E_0)\frac{\partial\Phi}{\partial k_1}(\cdot,k(s)), \frac{\partial\Phi}{\partial k_2}(\cdot, k(s))\Big>.
\end{equation}

To solve \eqref{XX0007}  for the partial derivatives of $f_0$ on $\hat \gamma$ up to order $l$, we may assume  that $f_1$ vanish to order $l-2$ at all points of $\hat\gamma$.Thus,  the right hand side of \eqref{XX0007} is no longer dependent on $f_1$. Since the coefficient $C$ depends  on all the partials  of $f_0$ up to order $l-1$, we can solve \eqref{XX0007} recursively. By linearity  the  solution exists for $all$ $s$.
Thus, it suffices to prescribe $\partial_y^\alpha f_0(y(0))$ for  $\vert \alpha\vert =l$,  to get the 
${\rm l}^{{\rm th}}$ order partial derivatives of $f_0$ on the whole curve $\hat\gamma$. Therefore,  we may assume that $f_0$ is chosen so that for given $l\in \mathbb N$ the left hand side of \eqref{XX00007} vanishes to order $l$ on $\hat\gamma$, i.e., 
\begin{equation}\tag{31}\label{XXXX7}
 \Big <\mathbb K(f_0\Phi), \Phi\Big>={\mathcal O}\Big(d(y, \Pi_{ y}\hat \gamma)^{l+1}\Big).
 \end{equation}
 This completes the construction of $m_0$.

 \begin{remark}  To compute  the other terms $m_1, m_2, \cdots,$ we proceed as follows.
 Set
 \begin{equation}\tag{32}\label{XXXX8}
 m_1^\perp=\frac{\hbar^2}{2m}(H_0-E_0)^{-1}\Big(\mathbb Km_0-\langle \mathbb Km_0,\Phi\rangle \Phi\Big)=\frac{\hbar^2}{2m}(H_0-E_0)^{-1}\Big(\mathbb Km_0\Big) ,
 \end{equation}
 where $(H_0-E_0)^{-1}$  denotes the inverse which maps the orthogonal complement of  of $\Phi$  in $L^2(\mathbb T)$ into itself. According to \eqref{XXXX05} we need only to compute $f_1$.
  We recall that $f_1$ needs  to  vanish to order $l-2$ at all points of $\hat\gamma$, and $G$ satisfies \eqref{XX00012}. Combining this with \eqref{XX05},  \eqref{XX00007}, \eqref{XXXX7} and \eqref{XXXX8} we get
    \begin{equation}\tag{33}\label{XXXXXX10}
  c_1={\mathcal O}\Big(d(y, \Pi_{ y}\hat \gamma)^{l+1}\Big)
  \end{equation}
  uniformly on $x\in \mathbb T$.
Like $T_1$ ,  $T_2$ can be solved for $m_2=f_2\Phi+m_2^\perp$   if and only
  \begin{equation}\tag{34}\label{XX00010}
 \Big <\mathbb K(f_1\Phi), \Phi\Big>_{L^2(\mathbb T)}= \frac{2m}{\hbar^2} f_2(y)G(y)-\Big <\mathbb K(m_1^\perp), \Phi\Big>_{L^2(\mathbb T)}- \Big <\Delta_{ y}(m_0), \Phi\Big>_{L^2(\mathbb T)}.
 \end{equation}
 The left side of the above equality  is identical to that of \eqref{XX00007} with $f_1$  instead of $f_0$.  Therefore, as in the proof of 
 \eqref{XX0009}, 
  equation \eqref{XX00010} gives
     \begin{equation}\tag{35}\label{XX00011}
     L_1(y):=
      \end{equation}
    $$ \frac{\partial E}{\partial  k}\cdot \frac{\partial f_1}{\partial  y}+ {\mathcal A}(y) f_1( y)
-\frac{2m}{i\hbar^2 } f_2(y)G(y) +\frac{1}{i}\Big[\Big <\mathbb K(m_1^\perp), \Phi\Big>_{L^2(\mathbb T)}+ \Big <\Delta_{ y}(m_0), \Phi\Big>_{L^2(\mathbb T)}\Big]=0. 
 $$
   Notice that the terms inside $\Big[\cdots \Big]$  vanish to order $l-2$ on $\hat\gamma$.   This follows from  the definition of $m_0$ and $m_1^\perp$ and the fact   $f_0$ vanishes to order $l$ on $\hat\gamma$.
 
Remembering that   we want to  construct   $f_1$ satisfying \eqref{XX00011} with $f_1(y)={\mathcal O}\Big(d(y, \Pi_{ y}\hat \gamma)^{l-2}\Big)$. 

As in  the proof of \eqref{XX0007} (with $l-2$ instead of $l$), equation \eqref{XX00011} shows that   for any multi-index $\alpha$ of length $l-2$ the equality  $\partial_y^\alpha  L_1(y)=0$ along $\hat\gamma$  yields 
a  linear o.d.f.  with an inhomogeneous term, $F$,  involving $m_0$ and $m_1^\perp$ :
 \begin{equation}\tag{36}\label{XXXX0009}
\frac{\partial E}{\partial  k}\cdot \frac{\partial }{\partial  y}(\partial_y^{\alpha}f_1)+ {\mathcal A}(y) \partial_y^\alpha f_1( y)+C_1(f_1,\cdots,\partial_y^{\beta}f_1, f_2,\cdots\partial_y^{\beta'}f_2)
=F.
\end{equation}
We get a linear differential equation for the  $\vert \alpha\vert$-order partial derivates of $f_1$ on $\hat \gamma$ with an inhomogeneous term previously determined. We can now proceed analogously to the proof of $f_0$.

 \end{remark}
 
 Let $\hat \gamma=\{(\hat y(s), \hat p(s));  s\in [0, T]\}$ be the null bicharacteristic of  the classical Hamiltonian $\hat H$ with initial date $(\hat y(0), \hat p(0))$ satisfying $\hat H (\hat y(0), \hat p(0))=E_0$, and let $ \Pi_{\hat y} \hat \gamma$  denote  the projection of  $\hat \gamma$ on the $\hat y$-space. Let $\phi$ the phase given by 
 Proposition 1. Next, we use the   above construction for both  $m_0$ and $m_1$ with $l=0$, $f_0(\hat y(0))=1$ and $f_1(\hat y(0))=0$. Let $\tilde \Omega \subset\mathbb R^2$ be a small neighborhood of $ \Pi_y\hat\gamma$, and let $f\in C^\infty_0(\tilde \Omega)$ be equal one near $ \Pi_y\hat\gamma$. Put
 $$\tilde u(\hat y,x_3, \epsilon)=e^ {{i}(\phi(\hat y)/\epsilon +x_3 k_3)} f(\hat y)\Big(m_0(\frac{\hat y}{\epsilon }, x_3,  \hat  y)+\epsilon m_1(\frac{\hat y}{\epsilon }, x_3,  \hat  y)\Big).$$
  A small enough $\tilde \Omega$ can be selected to ensure that the function $f$ is well defined on the support of $\phi$.

We state our main result of this section  as follows :
  
    \begin{theorem}\label{Prin1}
        Assume  \eqref{ASSU} and \eqref{ASSU2}. There exists $\epsilon_0$ small enough such that the preceding construction gives  an approximate eigenfunction $\tilde u$  satisfying :
        \begin{itemize}
        \item  $\vert (\tilde P-E_0)\tilde u (\hat y, x_3, \epsilon)\vert \leq C \epsilon^{3/2},\,\, (\hat y, x_3)\in \mathbb R^3.$
        \item $x\rightarrow m_j(x,  \hat  y)$ is periodic.
        \item  $$\Vert (\tilde P-E_0)\tilde u\Vert_{L^2(\mathbb R^2_{y_1, y_1})}={\mathcal O}(\epsilon ^{3/2})\Vert \tilde u\Vert_{L^2(\mathbb R^2_{y_1, y_1})}.$$
        uniformly on $\epsilon\in ]0,\epsilon_0[$.
        \item Along $\hat\gamma$, we have
       \begin{equation}\tag{37}\label{XXXXXXX7}
       \tilde u(\hat y(t),x_3, \epsilon)=e^ {{i}(c(t)/\epsilon-\theta(t) +x_3 k_3)} \sqrt{\frac{{\rm det}(Y(0))}{{\rm det}(Y(t))}} \Phi(\frac{\hat y(t)}{\epsilon}, x_3,  K(\hat  y(t))+{\mathcal O}(\epsilon),
       \end{equation}
      where
      $$c(t)=\int_0^t \hat p(s) d\hat y(s),\,\,\, \theta(t)=\int_0^t  \dot \theta_b+\dot\theta_{rw}\,  ds.$$
        \end{itemize}

 \end{theorem}
 
 \begin{remark}
 
 Our solution in \eqref{XXXXXXX7} involves the square root of the complex number 
 ${\rm det}(Y(s))$, and will be sensitive to the number of times the phase of  ${\rm det}(Y(s))$ wraps around the origin as $s$ goes from zero to $T$. If this winding number (or Maslov index) is called 
 $N_M$, we have by Cauchy's argument  principle :
 $$N_M=\frac{1}{2\pi i}\int_0^T \frac{\partial_s {\rm det }(Y(s))}{{\rm det}(Y(s))} ds.$$
 Therefore,
   \begin{equation}\tag{38}\label{XXXXX1111}
   \sqrt{\frac{{\rm det}(Y(0))}{{\rm det}(Y(T))}}=e^{-i\Theta_M},
   	\end{equation}
   	where $\Theta_M :=N_M\pi$ is the Maslov phase.

 Notice that, the phases in real, time-dependent WKB theory, i.e., the Maslov indices, change discontinuously at caustics and are not always easy to determine, whereas in the Gaussian beam method the phase is found by continuously following the quantity 
 ${\rm det}(Y(0))/{\rm det}(Y(t))$ and its square root along the orbit.

 \end{remark}

 \begin{proof}
   From \eqref{XX00012}, \eqref{XXXXX1} and \eqref{XXXXXX10},  we deduce
    \begin{equation}\tag{39}\label{XXXXX11}
    	c_0(\hat x, x_3, \hat y)={\mathcal O}\Big(d(y, \Pi_{ y}\hat \gamma)^{3}\Big)\,\,\, c_1(\hat x, x_3, \hat y)={\mathcal O}\Big(d(y, \Pi_{ y}\hat \gamma)\Big),
    \end{equation}
 uniformly on $x\in \mathbb T$. By Proposition 1, we have $\Im \phi\geq C d(\hat y,   \Pi_{\hat y} \hat \gamma)^2$. Hence
 $$d(\hat y,   \Pi_{\hat y} \hat \gamma)^m e^{i\phi/\epsilon}={\mathcal O}(\epsilon^{m/2}),$$
 which together with \eqref{XX03} and \eqref{XXXXX11}  produce the first three statements.

  By \eqref{XX0007}, we have
    \begin{equation}\tag{40}\label{XXXXXX7}
   \frac{d}{d s} \Big [f_0(y(s))\Big]+\frac{1}{2}\Big(\partial_{ y}\cdot \partial_{ k}E(\phi_{y}+\mu A( y), k_3)\Big)_{\vert  y(s)} f_0(y(s))+i\Big(\dot \theta_b+\dot\theta_{rw}\Big) f_0( y(s))=0.
   \end{equation}
We recall the $C(f_0,\partial_yf_0,\cdots,\partial_y^{\beta}=0, $ when $\alpha=0$. Next, by definition of $M(s)$, $B(s)$ and $C(s)$, we have : 
 $$ \partial_y\cdot \partial_kE(\phi_y+\mu A(y))|_{y(s)}={\rm tr}\Big(C(s)M(s)+B(s)\Big).$$
 From the first equation in \eqref{A33}, we deduce that $C(s)M(s)+B(s)=\dot Y(s) Y^{-1}(s)$. Combining this with the standard equality  $$ {\rm tr}(\dot Y(s)Y^{-1}(s))={d\over ds} \ln({\rm det}(Y(s))$$
 we obtain
$$(1/2)\partial_y\cdot \partial_kE(\phi_y+\mu A(y))|_{y(s)}=(1/2){d\over ds}\ln({\rm det}(Y(s)),$$ 
which together with \eqref{XXXXXX7} yields 
 \begin{equation}\tag{41}\label{TT001}
   \frac{d}{d s} \Big [f_0(\hat y(s))\Big]+\frac{1}{2}{d\over ds}\ln({\rm det}(Y(s)) f_0(\hat y(s))+i\Big(\dot \theta_b+\dot\theta_{rw}\Big)f_0(\hat y(s))=0.
 \end{equation}
Now \eqref{XXXXXXX7} follows from \eqref{XX02}, \eqref{XX03}, \eqref{Eikon} and \eqref{TT001}. We recall that $f_0(\hat y(0))=1$.

 \end{proof}

\section{Implications}


 \subsection{Generalized Onsager Relation.}
 
 { In physical space there  is a motion in the $y_3$-axis with velocity $\dot y_3(s)=\frac{\partial E}{\partial k_3}(k(s))$. Therefore,  the orbits $\hat\Gamma(k_3)$ of the
full hamiltonian in  \eqref{EM}, \eqref{EM0} are helical,  and do not support quasimodes, but their  projections  onto pseudo-momentum  produce resonances (called "magnetic energy levels" in the physics literature). {
 Onsager's key observation was that the magnetic energy levels determine $S(k_3)$ when it is extremal. Here we deduce that from the "resonance condition" that  the phase of the beam must increase by an integer  multiple of $2\pi$ when one goes around $\hat \gamma$. This means that $ \phi(y(s))/\epsilon+\Theta(s)$  increases by a multiple of  $2\pi$, where $\Theta(s)$   represents the combined contributions of the Berry, Wilkinson-Rammal  and Maslov phase shifts
 (see \eqref{Phases},  \eqref{RW} and Remark 3.)}

{ Let us compute the change in $\phi$ around  the periodic orbit  $\hat\gamma=\hat \gamma(k_3, E_0, c)$.  We denote $T$ its period.
From    \eqref{EQ},  $\phi_y(y(s))=p(s)$.  Since  $\dot p_2(s)=0$, it follows that
\begin{equation}\tag{42}\label{ST}
\dot \phi(y(s))=p(s)\cdot \dot y(s)=p_1(s)\dot y_1(s)+p_2(s) \dot y_2(s)= \frac{1}\mu p_1(s) [\dot p_2(s)+\mu \dot y_1(s)]+p_2(0) \dot y_2(s).
\end{equation}

 \vskip.2in
Recall that $\gamma (k_3,E_0)$ will become $\hat \gamma(k_3,E_0, c)$ when one substitutes $k_1=p_1$ and $k_2=p_2+\mu y_1$ and sets $y_2=-p_1/\mu +c$. Combining this with  the fact that
 $\int_0^T\dot y_2(s)ds=0, $ and using Green's  theorem we obtain from \eqref{ST}
$$\phi(y(T))-\phi(y(0))=\int_0^T p(s)\dot y(s)ds=\int_0^T  \frac{1}\mu p_1(s) [\dot p_2(s)+\mu \dot y_1(s)]ds$$
$$= \frac{1}{\mu}\int_{\hat \gamma} k_1dk_2=\frac{{\rm S}({k}_3)}{\mu}.$$}

Combining the above equality  with the resonance condition  and using  the fact that  that { $\mu \hbar =e$,}  we obtain
\begin{equation}\tag{43}\label{Onsa}
{ \frac{ \hbar  }{e\epsilon}{\rm S}({k}_3) -\Theta =2n\pi,}
\end{equation}

{ where $\Theta:=\Theta_b+\Theta_{rw}+\Theta_M$ now stands for the change in $\Theta(s)$ around $\hat\gamma$, i.e., 
$$
\Theta_b=-i\int_0^T \langle\Phi(\cdot,k(s)),\dot \Phi(\cdot, k(s))\rangle ds =-i\int_{\hat\gamma } \langle \Phi(\cdot,k),\partial_{k} \Phi(\cdot,k)\rangle dk,
$$
$$\Theta_{rw}=-\Im \int_0^T   \langle (H_0(k(s))-E)\frac{\partial\Phi}{\partial k_1}(\cdot,k(s)), \frac{\partial\Phi}{\partial k_2}(\cdot,k(s))\rangle \, ds.
$$
We recall that  $\dot \Phi(\cdot, k(s))=\dot k(s) \frac{\partial \Phi}{\partial k}(\cdot, k(s))$. Finally, according to \eqref{XXXXX1111}, 
$\Theta_M=0$
if $N_M$ is  even,  and $\Theta_M=\pi$ if $N_M$ is odd. Summing up, we obtain the generalized Onsager relation including the
Berry, Ramal-Wilkinson, and Maslov phases :


\begin{theorem}\label{Prin2} Assume (H1-2), and  let $S(k_3)$  denote the area in $k$-space   enclosed by $\gamma(k_3, E_0)$.
Therefore the quasi-classical quantization condition may be written in the form
\begin{equation}\tag{44}\label{ONSAGER}
\frac{ \hbar  }{e\epsilon}{\rm S}({k}_3)  =2\pi (n+\gamma)+\Theta_b+\Theta_{rw},
\end{equation}
where $\gamma=\frac{1}{2}$ or $0$.

\end{theorem}

\vskip.1in

\subsection*{Magnetic Oscillation.}
{ { The de Haas-van Alfven effect is associated with closed  curves $\gamma$ in pseudo-momentum space formed by intersecting the Fermi surface with planes
$k_3={k_3^0}$,  where the (constant) magnetic field is parallel to the $k_3$-axis. 
Not all choices for ${k}_3^0$ contribute to the de Haas-van Alfven effect. Only the "extremal" values
of ${k}_3^0$, i.e. those for which ${\partial S\over \partial k_3}=0$, contribute. { For each $n$ Onsager's relation \eqref{ONSAGER}  determines a magnetic energy level by giving a relation between $\epsilon$ and $S(k_3)$. }If $S(k_3^0) $ is extremal, as $k_3$ approaches $k_3^0$ smaller and smaller changes in $\epsilon$ are needed to satisfy Onsager's relation. Thus there is a peak in the density of magnetic energy levels at $k_3^0$.

 \vskip.4in
\section*{Appendix  A} \label{A}

   In this appendix, we recall some well  known facts about the non-linear Ricatti equation \eqref{Ricatti}. For reader convenience we sketch the proofs. To construct a solution of \eqref{Ricatti}, we start by choosing matrix solutions to the linear system
     \begin{equation}\tag{A1}\label{A33}
 \left \{
\begin{array}{lr}
\dot {\dot Y}= BY+CN \\
\\
\dot{N}=- AY-B^TN.
\end{array}
\right .
\end{equation}
   Since $A, B, C$  are well defined for all $s$, by linearity there exists a unique global solution $(Y(s), N(s))$ to the above system for any initial condition $(Y(0), N(0))$.
   
   Let $G_1(s)=(Y^1(s), N^1(s))$ and $G_2(s)=(Y^2(s), N^2(s))$ be two vectors solutions of \eqref{A33}. We recall that
    \begin{equation}\tag{A2}\label{A34}
    \sigma(G_1(s), G_2(s))=\sigma(G_1(0), G_2(0)).
    	\end{equation}
    	Since $\overline {G_2(s)}$ is also a solution of \eqref{A33}, the complexified form $\sigma_\mathbb C$ is also constant in $s$, i.e.,
\begin{equation}\tag{A3}\label{A35}
    \sigma_\mathbb C(G_1(s), G_2(s))= \sigma (G_1(s), \overline{G_2(s)})=\sigma_\mathbb C(G_1(0), G_2(0)).
    	\end{equation}
   
\begin{proposition} \label{Prop A1}
Let $M_0$ be a symmetric matrix such that $\Im M_0$ is positive definite on the orthogonal complement of $\dot y(0)=(\dot	y_1(0), \dot y_2(0))$, and $M_0\dot y(0)=\dot p(0)$. Let $(Y(s), N(s))$ be the solution of \eqref{A33} with initial condition $(Y(0), N(0))=(I, M_0)$. We have
\begin{itemize}
\item $(\dot y(s), \dot p(s))=(Y(s) \dot y(0), N(s) \dot y(0))$.
\item  	$Y(s)$ is invertible for all $s$.
\item $M(s):=N(s) Y(s)^{-1}$ is a solution of the Ricatti equation \eqref{Ricatti}.
\item $M(s)$ is a symmetric matrix for all $s\in \mathbb R$.
\item The matrix $\Im M(s)$ is positive definite on the orthogonal complement of $\dot y(s)=(\dot	y_1(s), \dot y_2(s))$.
\end{itemize}

\end{proposition}

\begin{proof}
The first claim follows from the fact that 
both  $(Y(s) \dot y(0), N(s) \dot y(0))$  and $(\dot y(s), \dot p(s))$ are 
vector solutions of  \eqref{A33} as well as  the fact  that $(\dot y(0), \dot p(0))=(Y(0) \dot y(0), N(0) \dot y(0))$.

Suppose that $Y(s) a=0$ for some $a\in \mathbb C^2$. Since $(x(s), \xi(s)):=(Y(s)a, N(s)a)$ is a vector solution of \eqref{A33}, it follows from the constancy of complexified form $\sigma_\mathbb C$ that
$$0=\sigma_{\mathbb C}((x(s), \xi(s)), (x(s),\xi(s)))=\sigma((x(0), \xi(0)), \overline {(x(0), \xi(0))})$$
$$= \overline{x(0)}\cdot \xi(0)-x(0)\cdot \overline{\xi(0)}=2i \overline{a}\cdot \Im (M_0) a.$$
By definition  $\Im M_0$ is positive definite on the orthogonal complement of $\dot y(0)$. Combining this  with the above equality we deduce  that   $a=\beta \dot y(0)$ for some $\beta \in \mathbb C$. Consequently,
$$0=Y(s)a=\beta Y(s) \dot y(0)=\beta \dot y(s),$$
where we have used the first item. This gives the second item since $\dot y(s)\not =0$ for all $s$, due to the   assumption (H1).

To deal with the third item, notice that $$\frac{d}{ds}(NY^{-1})=\dot N Y^{-1}-N Y^{-1} \dot Y Y^{-1},	$$
which together with \eqref{A33} yields
$$\frac{d}{ds}(NY^{-1})=- A-B^T (NY^{-1}\big)-(NY^{-1})BY-(NY^{-1})C(NY^{-1}).$$
Hence $M(s)=NY^{-1}$ satisfies \eqref{Ricatti}.

We pass to   the fourth statement. Let $Y^i(s)$ (resp $N^i(s)$), $i=1, 2,$ denote the column vectors of $Y(s)$ (resp. $N(s)$), and let $G_i(s)=(Y^i(s), N^i(s))$. By construction of $M(s)$, we have $M(s) Y^i(s)=N^i(s)$.

From \eqref{A34}, we have
$$\sigma(G_i(s), G_j(s))=\sigma(G_i(0), G_j(0))=Y^j(s)\cdot M(s) Y^i(s)-Y^i(s)\cdot M(s)Y^j(s)
$$
$$=Y^j(0)\cdot M(0) Y^i(0)-Y^i(0)\cdot M(0)Y^j(0).$$
The right hand side of the last equation equals zero, since $M(0)=M_0$ is symmetric. Therefore,  for all $s\in \mathbb R$,
$$Y^j(s)\cdot M(s) Y^i(s)=Y^i(s)\cdot M(s)Y^j(s).$$
Notice that $(Y^1(s), Y^2(s))$ is a basis in $\mathbb C^2$, since $Y(s)$ is invertible for all $s$. Combining this with the above equality we deduce that $M(s)$ is symmetric.

The proof for the final item is almost identical to the proof of the second one, and relies on the conservation of the complexified form $\sigma_{\mathbb C}$.

\end{proof}

\end{document}